\documentclass[letter,journal,final]{IEEEtran}  
\usepackage{arxiv}  
\usepackage{amsthm}  
\newcommand{\publicationversion}{Preprint 
}

\usepackage{ifpdf}
\usepackage{amsmath,amssymb,amsfonts}
\usepackage{url}

\usepackage{xcolor}
\usepackage{epsfig} 
\usepackage{epstopdf}
\usepackage{cite} 
\usepackage{url}
\usepackage{hyperref}
\hypersetup{
	colorlinks=true,
	linkcolor=blue,
	filecolor=blue,
	urlcolor=black,
	citecolor = blue,
}
\usepackage{easyReview} 
\allowdisplaybreaks

\newtheorem{thm}{Theorem}
\newtheorem{lem}{Lemma}

\newtheorem{assum}{Assumption}

\newtheorem{rem}{Remark}

\newtheorem{prob}{Problem}

\newcommand{\sth}{\textit{s.t.}}
\newcommand{\T}{^{\top}} 
\newcommand{\ie}{\textit{i.e.}}

\newcommand{\diag}{\operatorname{diag}}

\newcommand{\SO}{\operatorname{SO}}

\newcommand{\tr}{\operatorname{tr}}

\newcommand{\MPC}{\mathrm{MPC}}
\newcommand{\bu}{\mathbf{u}}

\title{MPC-Based Trajectory Tracking for a Quadrotor UAV with Uniform Semi-Global Asymptotic Stability Guarantees}
\author{Qian Yang, Miaomiao Wang, Abdelhamid Tayebi
\thanks{This work was supported in part by the National Natural Science Foundation of China under Grant 62403205, in part by the Interdisciplinary Research Program of HUST under Grant 2026JCYJ077, and in part by the National Sciences and Engineering Research Council of Canada under Grant NSERC-DG RGPIN-2026-06037. Q. Yang and M. Wang are with the School of Artificial Intelligence and Automation, Huazhong University of Science and Technology, Wuhan 430074, China (e-mail: {\tt\small mmwang@hust.edu.cn, hustyangqian@hust.edu.cn}). A. Tayebi is with the Department of Electrical Engineering, Lakehead University, Thunder Bay, ON P7B 5E1, Canada (e-mail: {\tt\small atayebi@lakeheadu.ca}). } }

\begin{document}

\twocolumn[ 
\begin{@twocolumnfalse}
\maketitle
\thispagestyle{empty}
\pagestyle{empty}
\begin{abstract}
   This paper proposes a model predictive trajectory tracking approach for quadrotor UAVs subject to input constraints. Our proposed approach relies on a hierarchical control strategy with an outer-loop feedback generating the required thrust and desired attitude, and an inner-loop feedback regulating the actual attitude to the desired one. For the outer-loop translational dynamics, the generation of the virtual control input is formulated as a constrained model predictive control problem with time-varying input constraints and a control strategy, endowed with uniform global asymptotic stability guarantees, is proposed. For the inner-loop rotational dynamics, a hybrid geometric controller is adopted, achieving semi-global exponential tracking of the desired attitude. Finally, we prove that the overall cascaded system is semi-globally asymptotically stable. Simulation results illustrate the effectiveness of the proposed approach. 
\end{abstract}
\end{@twocolumnfalse}]
\arxivthanks{This work was supported in part by the National Natural Science Foundation of China under Grant 62403205, in part by the Interdisciplinary Research Program of HUST under Grant 2026JCYJ077, and in part by the National Sciences and Engineering Research Council of Canada under Grant NSERC-DG RGPIN-2026-06037. Q. Yang and M. Wang are with the School of Artificial Intelligence and Automation, Huazhong University of Science and Technology, Wuhan 430074, China (e-mail: {\tt\small mmwang@hust.edu.cn, hustyangqian@hust.edu.cn}). A. Tayebi is with the Department of Electrical Engineering, Lakehead University, Thunder Bay, ON P7B 5E1, Canada (e-mail: {\tt\small atayebi@lakeheadu.ca}).}%
\copyrightfootnote{\copyrightNoticeAccepted{IEEE CDC 2026}}%

\section{Introduction}
Quadrotor unmanned aerial vehicles (UAVs) are widely used in many industrial and hobby applications, due to their high maneuverability and vertical takeoff and landing capability. Due to the underactuated nature of these systems, the hierarchical inner-outer loop control paradigm is a commonly adopted control scheme for the trajectory tracking problem for quadrotor UAVs. The outer-loop feedback deals with the translational dynamics, where a virtual control input is designed to regulate the position of the quadrotor UAV to the reference trajectory. Based on this virtual input, the thrust input and the desired orientation/attitude are extracted, and then the inner-loop feedback takes care of regulating the actual orientation to the desired one (see, for instance, \cite{Abdess_Tay_2010,Roberts_Tayebi_2011,Roberts_Tayebi_2013}). The outer-loop control feedback for trajectory tracking is mainly designed based on the classical PD or PID approaches \cite{Roberts_Tayebi_2013,lee2011geometric}. Various attitude control techniques for the inner-loop system have been developed in the literature using different attitude representations such as the unit-quaternion \cite{tayebi2006attitude,Tayebi_2008} and the special orthogonal group $\SO(3)$ \cite{lee2011geometric,wu2024geometric}. Due to the topological obstruction to global asymptotic stability (GAS) on $\SO(3)$, the traditional continuous time-invariant feedback cannot achieve GAS on $\SO(3)$(see, for instance, \cite{bhat2000topological}). To overcome this limitation, hybrid feedback control techniques, endowed with global asymptotic stability guarantees, have been proposed in the literature \cite{mayhew2013synergistic,Ber_Abdes_Tay_2018,Wang_Tayebi2022}. 

In practical applications, the trajectory tracking of quadrotor UAVs must take into consideration practical requirements such as actuator saturation. Existing methods usually guarantee these requirements by restricting the outer-loop control parameters. Model predictive control (MPC) is an attractive approach when it comes to handling input constraints. In \cite{andrien2024model}, MPC was employed for the outer-loop feedback design for a quadrotor UAV. The translational dynamics were dynamically extended to generate a twice-differentiable desired attitude to be used by the inner-loop controller. To ensure that the actual virtual input satisfies the prescribed time-varying constraint, a uniform conservative constraint was imposed on the augmented virtual input in the MPC formulation. To relax the conservative constraint conditions in \cite{andrien2024model}, the authors in \cite{izadi2024high} suggested updating the time-varying constraint at each sampling instant. Both approaches in \cite{andrien2024model,izadi2024high} achieve uniform almost global asymptotic stability (UAGAS) of the overall closed-loop system. However, neither approach allows the original time-varying constraint to be imposed directly on the augmented virtual input.

In this paper, we propose a hierarchical control framework for quadrotor trajectory tracking. In the outer loop, an MPC strategy is designed to achieve UGAS in the continuous-time sense. In particular, based on \cite{andrien2024model}, a scale factor is introduced so that the original time-varying constraints can be applied directly to the augmented virtual input in the MPC formulation. For the inner loop, a hybrid attitude controller is adopted to ensure semi-global exponential stability (SGES). Based on these results, uniform semi-global asymptotic stability (USGAS) of the overall cascaded closed-loop system is rigorously established.

\section{Preliminaries}
\subsection{Notations and Definitions}
Denote the sets of all real numbers and non-negative integers by $\mathbb{R}$ and $\mathbb{N}$, respectively. We denote by $\mathbb{R}^n$ the $n$-dimensional Euclidean space, and by $\mathcal{S}^{n}$ the set of unit vectors in $\mathbb{R}^{n+1}$. Let $I_n$ and $0_{m\times n}$ denote the $n$-by-$n$ identity matrix and the $m$-by-$n$ matrix with all zero elements, respectively. Let $e_i\in\mathbb{R}^3$ denote the standard $i$-th basis vector, which represents the $i$-th column of the identity matrix $I_3$. Define the $2$-norm of a vector $a\in\mathbb{R}^n$ as $\|a\|$, and the induced $2$-norm of a matrix $A\in\mathbb{R}^{m\times n}$ as $\|A\|_2$. For some symmetric matrix $P=P\T\succ 0$ in $\mathbb{R}^{n\times n}$, we define $\|x\|_P^2=x\T P x$ for a vector $x\in\mathbb{R}^n$. Moreover, let $\vert x\vert _\mathcal{A}$ denote the distance of a point $x\in\mathcal{M}$ to a closed set $\mathcal{A}\subset\mathcal{M}$.  We write $\lambda_{\max}(P)$ ($\lambda_{\min}(P)$) for the maximum (minimum) eigenvalue of $P$. For any vector $p\in \mathbb{R}^n$, we denote by $\underline{p}^i$ the $i$-th component of $p$. For any vector $x=[x_1,x_2,x_3]\T\in\mathbb{R}^3$, we define $(\cdot )^\times $ as a skew-symmetric operator given by $x^\times=[0,-x_3,x_2;x_3,0,-x_1;-x_2,x_1,0]$, such that for any $x,y\in\mathbb{R}^3$, the identity $x^\times y=x\times y$ holds, where $\times$ denotes the vector cross product operator, and define $\mathrm{vec}(\cdot)$ as the inverse operator of the map $(\cdot)^\times$, such that $\mathrm{vec}(x^\times)=x$. Define the map $\psi:\mathbb{R}^{3\times 3}\to \mathbb{R}^3$ by $\psi(P):=\frac{1}{2}\mathrm{vec}(P-P\T)$ for all $ P\in\mathbb{R}^{3\times 3}$. In addition, define $\overline{\mathrm{vec}}(\cdot):\mathbb{R}^{m\times n}\to\mathbb{R}^{mn}$ as the column-stacking operator, which maps a matrix to a vector by concatenating its columns in order. Let $\SO(3)$ denote the special orthogonal group of order three, given by $\SO(3):=\{R\in\mathbb{R}^{3\times 3}: RR\T=R\T R=I_3,\det(R)=1\}$. Denote the map $\mathcal{R}:\mathbb{R}\times \mathcal{S}^2\to \SO(3)$ by $\mathcal{R}(\theta,u):=I_3 +\sin(\theta)u^\times +(1-\cos(\theta))(u^\times)^2, \forall (\theta,u)\in\mathbb{R}\times \mathcal{S}^2$.

We introduce a differentiable saturation function $\sigma_\Delta:\mathbb{R}\to\mathbb{R}$ with $\Delta>0$, which satisfies the following properties:
\begin{enumerate}
    \item $\sigma_\Delta(0)=0$ and $\vert \sigma_\Delta(x)\vert \leq \Delta$ for all $x\in\mathbb{R}$;
    \item $x\sigma_\Delta(x)>0$ for all $ x\neq 0$;
    \item $\mathcal{D}\sigma_\Delta(x)= \frac{\partial \sigma_\Delta(x)}{\partial x}\in (0,1]$ is non-increasing on $(0,+\infty)$ and non-decreasing on $(-\infty,0)$. 
\end{enumerate}
The following lemma on the saturation function will be used throughout this paper.
\begin{lem}\label{lem_1}
    For all $a>0, b\in[0,1]$, and $x\in\mathbb{R}$, the following inequality holds:
    \begin{enumerate}
        \item $\sigma_\Delta(ax)f(x)\leq axf(x)$ with a function $f:\mathbb{R}\to\mathbb{R}$ satisfying $xf(x)>0,\forall x\neq 0$ and $f(0)=0$;
        \item $bx\sigma_\Delta(x)\leq x\sigma_\Delta(bx)$;
        \item $\Psi_\Delta(x):=\int_0^x \sigma_\Delta(s)ds\leq x\sigma_\Delta(x)$.
    \end{enumerate}
\end{lem}
\begin{proof}
    See Appendix \ref{adx_lem1}.
\end{proof}
\section{Problem Formulation}
Let $R\in \SO(3)$ denote the attitude of the body-fixed frame attached to a quadrotor UAV with respect to the inertial frame, and $\omega\in\mathbb{R}^3$ denote the angular velocity expressed in the body-fixed frame. Let $p,v\in\mathbb{R}^3$ denote the position and linear velocity of the quadrotor UAV expressed in the inertial frame, respectively. We consider the following dynamic model of a quadrotor UAV:
\begin{subequations}\label{eqn_model}
    \begin{align}
        \dot{p}&=v,\quad\quad \dot{v}=ge_3-TRe_3-Dv\\
        \dot{R}&=R\omega^\times, ~I_b\dot{\omega}=-\omega\times I_b\omega +\tau\label{eqn_model_inner_omega}
      \end{align}
\end{subequations}
where $g$ is the acceleration of gravity, $D=\diag(d_1,d_2,d_3)$ with $d_1,d_2,d_3>0$ denoting the mass-normalized rotor drag coefficients, $I_b=I_b\T\in\mathbb{R}^{3\times 3}$ is the constant inertia matrix, $\tau \in\mathbb{R}^3$ is the control torque input, and $T\in [0,T_{\max}]$ denotes the thrust normalized by mass limited by denoting the maximal thrust $T_{\max}>g$. 
The dynamic model \eqref{eqn_model} is adapted from \cite{andrien2024model} with \eqref{eqn_model_inner_omega} written in a more standard form, where the difference between these two models can be absorbed into the feedforward term and does not materially affect the subsequent controller design.

The objective of this work is to design $T$ and $\tau$ for system \eqref{eqn_model} to track a given continuous time-varying reference $r(t)\in\mathbb{R}^3$, which satisfies the following assumption:
\begin{assum}\label{assum_1}
    Suppose that $r(t)$ is four times differentiable, and its derivatives $\dot{r}(t),\ddot{r}(t),r^{(3)}(t)$ and $r^{(4)}(t)$ are uniformly bounded with known bounds. In addition, there exist positive constants $\delta_r$ and $\delta_{rz}$ satisfying $\|D\dot{r}+\ddot{r}\|\leq \delta_r<T_{\max}-g$ and $e_3\T (D\dot{r}+\ddot{r})\leq\delta_{rz}<g$.
\end{assum}

Define the state tracking errors as $\tilde{p}:=p-r,\tilde{v}:=v-\dot{r}$, and then the error dynamics are given by 
\begin{subequations}\label{eqn_error_outer}
    \begin{align}
        \dot{\tilde{p}}&=\tilde{v}\\
        \dot{\tilde{v}}&=ge_3-TRe_3-\ddot{r}-D\dot{r}-D\tilde{v}=\mu-D\tilde{v}
    \end{align}
\end{subequations}
with $\mu=ge_3-TRe_3-\ddot{r}-D\dot{r}$. The variable $\mu$ can be viewed as a virtual control input for the error dynamics \eqref{eqn_error_outer}, such that $\tilde{p}$ and $\tilde{v}$ can be regulated to zero by appropriately designing the virtual input $\mu$. 
Let $\mu_d$ denote the desired virtual control input generated from the outer-loop position tracking control feedback, which can be realized in terms of the thrust $T$ and the desired attitude $R_d$ as follows:
\begin{align}\label{eqn_mud}
\mu_d=ge_3-TR_de_3-\ddot{r}-D\dot{r} 
\end{align}
where $R_d$ denotes the desired attitude to be tracked by the inner-loop control feedback for the actual attitude $R$ via the control torque $\tau$. Therefore, given the desired virtual input $\mu_d$ from the outer-loop control feedback, the thrust magnitude $T$ and the desired attitude $R_d$ can be extracted from \eqref{eqn_mud}. However, equation \eqref{eqn_mud} determines only the third column of $R_d$ (\ie, $R_de_3$), so the extraction of $R_d$ from $\mu_d$ is intrinsically non-unique. Furthermore, the extraction map, as proposed in \cite{wu2024geometric}, may have a singularity issue when $\mu_d$ is not designed properly. Hence, besides the physical constraints induced by $T_{\max},\dot{r}$, and $\ddot{r}$, the desired virtual input $\mu_d$ must satisfy an additional admissibility condition to guarantee a well-defined extraction of $R_d$. These conditions impose time-varying constraints on the virtual input $\mu_d$, which will be explicitly characterized in Section \ref{sec_III}.
 
Hence, the control design is decomposed into the following outer-loop and inner-loop control design problems:
\begin{prob}[Outer-loop control]\label{prob_outer}
    Design an MPC-based control scheme for the desired virtual control input $\mu_d$ under appropriate admissibility constraints such that the origin of system \eqref{eqn_error_outer} is UGAS.
\end{prob}

\begin{prob}[Inner-loop control]\label{prob_inner}
    Based on the desired attitude $R_d$ extracted from $\mu_d$, design a control torque $\tau$ such that $R$ tracks $R_d$ with uniform semi-global asymptotic stability guarantees for the overall closed-loop system.
\end{prob}

\section{Control System Design}\label{sec_III}
\subsection{Thrust, Desired Attitude and Input Constraints}
For convenience, we represent the desired attitude $R_d\in\SO(3)$ in terms of its column vectors as $R_d =[\mathbf{x}_d,\mathbf{y}_d,\mathbf{z}_d]$, where $\mathbf{x}_d,\mathbf{y}_d,\mathbf{z}_d \in\mathbb{R}^3$ constitute a right-handed orthonormal frame. From \eqref{eqn_mud}, the total thrust magnitude is given by
\begin{align}\label{eqn_thrust}
    T = \| ge_3 - D\dot{r}- \ddot{r}  -\mu_d\|
\end{align}
and the $z-$axis of the desired attitude is determined by the normalized thrust direction
\begin{align}\label{eqn_zd}
    \mathbf{z}_d = R_d e_3 = \frac{1}{T}(ge_3 - D\dot{r}- \ddot{r}  -\mu_d).
\end{align}
To specify the vehicle heading, we introduce a twice differentiable desired yaw angle $\psi_d(t)$, whose first and second derivatives are uniformly bounded. This angle can be chosen to satisfy task-dependent orientation requirements, and its corresponding heading vector is defined as
\begin{align}\label{eqn_xb}
    \mathbf{x}_b = \begin{bmatrix}
        \cos\psi_d & \sin\psi_d & 0
    \end{bmatrix}\T.
\end{align}
Then, the first column of $R_d$ is given by
\begin{align}\label{eqn_xd}
    \mathbf{x}_d = \frac{\pi_z \mathbf{x}_b}{\|\pi_z\mathbf{x}_b\|}
\end{align}
with $\pi_z = I_3-\mathbf{z}_d\mathbf{z}_d\T$, denoting the orthogonal projection on the plane orthogonal to $\mathbf{z}_d$. To avoid the singularity that $\pi_z\mathbf{x}_b = 0$, \ie, $\mathbf{x}_b$ is parallel to $\mathbf{z}_d$, $\mu_d$ is required to satisfy $e_3\T\mu_d<g-e_3\T D\dot{r}-e_3\T \ddot{r}$, which can be ensured by imposing
\begin{align}\label{eqn_constraint_1}
    \|\mu_d\|\leq g-e_3\T D\dot{r}-e_3\T\ddot{r}-\epsilon :=\mathcal{B}_1(\dot{r},\ddot{r})
\end{align}
with $0<\epsilon<g-\delta_{rz}$. Then, one has the following lemma:
\begin{lem}
     Under Assumption \ref{assum_1} and constraint \eqref{eqn_constraint_1}, the total thrust \eqref{eqn_thrust} is guaranteed to be strictly positive.
\end{lem}
\begin{proof}
    Under Assumption \ref{assum_1}, the reference trajectory $r(t)$ satisfies $g-e_3\T(D\dot{r}+\ddot{r})>0$. Then, in view of \eqref{eqn_thrust} and \eqref{eqn_constraint_1}, one can obtain 
    \begin{align*}
        T&\geq \vert g - e_3\T D\dot{r}-e_3\T\ddot{r}-e_3\T \mu_d\vert\\& \geq \vert g-e_3\T D\dot{r}-e_3\T\ddot{r}\vert -e_3\T \mu_d\\
        &\geq g-e_3\T D\dot{r} - e_3\T\ddot{r} -\|\mu_d\|\\
        &\geq \epsilon >0.
    \end{align*}
    Consequently, the thrust $T$ remains strictly positive.
\end{proof}

Due to actuator saturations, the thrust \eqref{eqn_thrust} has to satisfy $T=\|ge_3-\ddot{r}-D\dot{r}-\mu_d\|\leq T_{\max}$. A sufficient conservative condition is 
\begin{align}\label{eqn_constraint_2}
    \|\mu_d\|\leq T_{\max}-\|ge_3-\ddot{r}-D\dot{r}\|:=\mathcal{B}_2(\dot{r},\ddot{r}).
\end{align}
Combining with \eqref{eqn_constraint_1} and \eqref{eqn_constraint_2}, the time-varying bound on the magnitude of $\mu_d$ is given by
\begin{align}\label{eqn_tv_constraint}
    \|\mu_d\|\leq \mathcal{B}(\dot{r},\ddot{r}):=\min\left\{\mathcal{B}_1(\dot{r},\ddot{r}),\mathcal{B}_2(\dot{r},\ddot{r})\right\}.
\end{align}

Similar to \cite{wu2024geometric}, given $\mu_d$ satisfying \eqref{eqn_tv_constraint}, the desired attitude is given by
\begin{align}\label{eqn_desired_attitude}
    R_d = [\mathbf{x}_d,\mathbf{z}_d\times\mathbf{x}_d,\mathbf{z}_d].
\end{align}
Moreover, for the design of the inner-loop control feedback, the desired angular velocity $\omega_d$ and angular acceleration $\dot{\omega}_d$ are typically required (See, for instance, \cite{lee2010geometric}). Note that $\omega_d$ can be determined from the kinematics $\dot{R}_d = R_d \omega_d^\times$ and then $\dot{\omega}_d$ depends on $\ddot{R}_d$, where $R_d$ is extracted from \eqref{eqn_thrust}-\eqref{eqn_xd} and \eqref{eqn_desired_attitude}. The explicit expressions for $\omega_d$ and $\dot{\omega}_d$ are provided in Appendix \ref{apx_omega_d}. Hence, the desired virtual input $\mu_d$ to be designed in the outer-loop subsystem must be twice differentiable.

\subsection{MPC-based Feedback Design for Outer-loop Subsystem}
To ensure that the virtual input \eqref{eqn_mud} satisfies the time-varying constraints as in \eqref{eqn_tv_constraint}, the works in \cite{Abdess_Tay_2010,Roberts_Tayebi_2011,Roberts_Tayebi_2013,yang2025quaternion} impose additional restrictions on the gain parameters of the outer-loop feedback controllers, which in turn degrade the control performance. To avoid such additional gain parameter restrictions, MPC will be employed in this subsection due to its capability for explicit constraint handling.

In order to provide twice differentiable desired virtual input, inspired by the work \cite{andrien2024model}, we extend the system \eqref{eqn_error_outer} as follows:
\begin{subequations}\label{eqn_exCL}
    \begin{align}
        \dot{\tilde{p}}&=\tilde{v},\quad\quad\quad\quad\quad\quad \dot{\tilde{v}}=-D\tilde{v}+\mu_d\\
        \dot{\mu}_d&=-\frac{1}{\gamma}(\mu_d-\alpha\eta),~~\dot{\eta}=-\frac{1}{\gamma}(\eta-\alpha u)
    \end{align}
\end{subequations}
where $\gamma>0,\tilde{p},\tilde{v},\mu_d,\eta\in\mathbb{R}^3,u\in\mathbb{R}^3$ is the augmented virtual input and $\alpha$ is a scale factor. Compared with \cite{andrien2024model}, the scale factor $\alpha$ is introduced to allow the time-varying constraint of $\mu_d$ to be applied directly to the augmented virtual input, rather than to its lower bound. The selection of $\gamma$ and $\alpha$ will be discussed in Lemma \ref{lem_gamma_alpha}. 

Similar to \cite{andrien2024model,mueller2013model}, due to the diagonal form of the matrix $D$, the vector system \eqref{eqn_exCL} is decomposed into three independent scalar subsystems to facilitate the controller design:
\begin{subequations}\label{eqn_separate_system}
    \begin{align}
        \dot{\tilde{\underline{p}}}^i&=\tilde{\underline{v}}^i,\quad\quad\quad\quad\quad\quad ~\dot{\tilde{\underline{v}}}^i=-d_i\tilde{\underline{v}}^i+\underline{\mu_d}^i\\
        \dot{\underline{\mu}}_d^i&=-\frac{1}{\gamma}(\underline{\mu_d}^i-\alpha\underline{\eta}^i),~~\dot{\underline{\eta}}^i=-\frac{1}{\gamma}(\underline{\eta}^i-\alpha \underline{u}^i)\label{eqn_separate_system_4}       
    \end{align}
\end{subequations}
with $i\in\{1,2,3\}$. To satisfy the constraint in $\eqref{eqn_tv_constraint}$, a more conservative constraint is imposed on each  
$\mu_d^i$ as
\begin{align}\label{eqn_constraint}
    \vert \underline{\mu_d}^i\vert \leq \underline{\mathcal{B}}(\dot{r},{\ddot{r}}):=\frac{1}{\sqrt{3}}\mathcal{B}(\dot{r},\ddot{r}).
\end{align}
Let $x^i=[\tilde{\underline{p}}^i,\tilde{\underline{v}}^i,\underline{\mu_d}^i, \underline{\eta}^i]\T\in\mathbb{R}^4$. Applying exact discretization for system \eqref{eqn_separate_system} with a sampling period $h>0$ under a zero-order hold (ZOH) on the input (\ie, $\underline{u}^i(t)=\underline{u}^i(t_k), t\in[t_k,t_{k+1})$ and $t_k=kh,k\in \mathbb{N}$), one obtains the following linear discrete-time subsystem:
\begin{align}\label{eqn_discrete_system}
    x_{k+1}^i=\underline{A}^ix_k^i+\underline{B}^i\underline{u}_k^i
\end{align}
where $x_k^i=x^i(t_k), \underline{u}_k^i=\underline{u}^i(t_k)$ with $k\in \mathbb{N}$ for the sake of simplicity, and the system matrices $\underline{A}^i$ and $\underline{B}^i$ are given by
\begin{align}\label{eqn_calculate_AB}
    \underline{A}^i=\exp{(\underline{A}_0^i h)},\underline{B}^i= \int_0^h \exp(\underline{A}_0^i \tau)\underline{B}_0^i d\tau
\end{align}
with
\begin{align}\label{eqn_A0B0}
    \underline{A}_0^i=\begin{bmatrix}
        0 & 1 & 0 & 0 \\ 0 & -d_i & 1 & 0 \\ 0 & 0 & -\frac{1}{\gamma} & \frac{\alpha}{\gamma} \\ 0 & 0 & 0 & -\frac{1}{\gamma}
    \end{bmatrix},\underline{B}_0^i=\begin{bmatrix}
        0 \\ 0 \\ 0 \\ \frac{\alpha}{\gamma}
    \end{bmatrix}.
\end{align}
Since the three scalar subsystems have the same structure, the index $i$ will be omitted from here.  The time-varying bound $\underline{\mathcal{B}}(\dot{r},\ddot{r})$ on $\underline{\mu_d}$ in each period can be bounded by a positive constant given by
$ \Delta_k:=\inf_{kh\leq t<(k+1)h}\underline{\mathcal{B}}(\dot{r},{\ddot{r}}) >0, \forall k \in \mathbb{N}$.   
Then, by applying this constraint directly on the control input $\underline{u}_k$ within each sampling period, it can be ensured that, for each $\vert \underline{\mu_d}(0)\vert\leq\Delta_0 $ and $\vert \underline{\eta}(0)\vert \leq \Delta_0$, the virtual input $\underline{\mu_d}$ satisfies the time-varying constraint \eqref{eqn_constraint}. This is guaranteed by the following lemmas:
\begin{lem}\label{lem_3}
Under Assumption~\ref{assum_1}, the input constraint \eqref{eqn_constraint} is Lipschitz continuous with respect to time, \ie
\begin{align*}
        \vert \underline{\mathcal{B}}(\dot{r}(t_1),\ddot{r}(t_1))- \underline{\mathcal{B}}(\dot{r}(t_2),\ddot{r}(t_2))\vert\leq \bar{L}\vert t_1-t_2\vert,\forall t_1,t_2>0,    
\end{align*}
where $\bar{L} := (\bar{d}L_1 +L_2)/\sqrt{3}$ with $L_1 :=\sup_{t\geq 0}\|\ddot{r}(t)\|$ and $L_2:=\sup_{t\geq 0}\|r^{(3)}(t)\|$, and $\bar{d} = \max\{d_1,d_2,d_3\}$.
\end{lem}
\begin{proof}
    See Appendix \ref{apx_lem3}.
\end{proof}
\begin{lem}\label{lem_gamma_alpha}
    Consider system \eqref{eqn_separate_system}, and let the parameter $\gamma$ and the scale factor $\alpha$ satisfy the following inequalities:
    \begin{align}\label{eqn_lem3}
            0<\gamma<\frac{\Delta}{\bar{L}},\quad 0<\alpha\leq\frac{\beta-e^{-h/\gamma}}{1-e^{-h/\gamma}}
    \end{align}
    with $\Delta=\inf_{t\geq 0} \underline{\mathcal{B}}(\dot{r},{\ddot{r}})\geq \frac{1}{\sqrt{3}}\min\{T_{\max}-g-\delta_r,g-\delta_{rz}-\epsilon\}$ and $\beta=\Delta/(\Delta+\bar{L}h)$. Then, given $\vert\underline{\mu_d}(t_k)\vert\leq \Delta_k,\vert \underline{\eta}(t_k)\vert\leq \Delta_k$ and $\vert \underline{u}_k\vert\leq \Delta_k$, the following statements hold:
    \begin{enumerate}
        \item $\vert \underline{\mu_d}(t)\vert \leq \Delta _k$ and $\vert \underline{\eta}(t)\vert \leq \Delta_k$ for all $t\in[t_k,t_{k+1})$;
        \item $\vert \underline{\mu_d}(t_{k+1})\vert \leq \Delta_{k+1}$ and $\vert \underline{\eta}(t_{k+1})\vert \leq\Delta_{k+1}$.
    \end{enumerate}
\end{lem}
\begin{proof}
    See Appendix \ref{apx_lem4}.
\end{proof} 

Given the parameter $\gamma$ and the scale factor $\alpha$ according to \eqref{eqn_lem3} as per Lemma \ref{lem_gamma_alpha}, one has $\vert \underline{\mu_d}(t)\vert\leq \underline{\mathcal{B}}(\dot{r},{\ddot{r}})$ for all $t \geq 0$ from the initial conditions $\vert \underline{\mu_d}(0)\vert\leq \Delta_0, \vert \underline{\eta}(0)\vert\leq \Delta _0$ and the input constraint $\vert \underline{u}_k\vert \leq \Delta_k$. This guarantees that the desired attitude extraction in \eqref{eqn_desired_attitude} is feasible and actuator saturations are satisfied for all $t\geq 0$.
\begin{rem}
    Compared with \cite{izadi2024high}, the introduction of the scale factor allows the constraint on the virtual input $\underline{\mu_d}$ to be directly applied to the augmented virtual input $\underline{u}$, thereby avoiding the additional constraint calculations required within each sampling period. It should be noted, however, that in order to guarantee $\underline{\mu_d}(t_k)\leq \Delta_k$ and $\underline{\eta}(t_k)\leq \Delta_k$ at the beginning of every sampling instant $t_k$, the scale factor $\alpha$ satisfies $\alpha<1$. Consequently, the virtual input $\underline{\mu_d}$ cannot fully saturate its admissible constraint $\vert\underline{\mu_d}\vert\leq \underline{\mathcal{B}}(\dot{r},\ddot{r})$. A similar effect also appears in \cite{izadi2024high} due to the additional constraint computations introduced there.
\end{rem}

Since UGAS under MPC is generally difficult to establish for open-loop non-asymptotically stable systems, we next transform the outer-loop subsystem \eqref{eqn_separate_system} into a form that explicitly reveals its stability structure. In particular, the system is decomposed into a marginally stable subsystem and an asymptotically stable subsystem. This decomposition allows the subsequent MPC design and stability analysis to focus primarily on the marginally stable subsystem.

Therefore, some preliminary computations are carried out. Transforming matrix $\underline{A}_0$ to its Jordan canonical form, one can obtain:

\begin{align*}
    P^{-1}\underline{A}_0P&=\begin{bmatrix}
        -d & 0 & 0 & 0 \\ 0 & -\frac{1}{\gamma} & 1 & 0 \\ 0 & 0 & -\frac{1}{\gamma} & 0 \\ 0 & 0 & 0 & 0
    \end{bmatrix}:=J_0
\end{align*}
with 
\begin{align*}
    P=\begin{bmatrix} 1 & -\gamma^2 & -\gamma^3+\frac{\gamma^3}{\gamma d-1} & 1\\ -d & \gamma & \frac{\gamma^2}{1-\gamma d} & 0\\ 0 & \gamma d-1 & 0 & 0\\ 0 & 0 & \frac{\gamma(\gamma d-1)}{\alpha} & 0 \end{bmatrix}.
\end{align*}
In view of \eqref{eqn_calculate_AB}, the system matrices $\underline{A}$ and $\underline{B}$ can be written as $\underline{A}=P\exp(J_0 h) P^{-1}:=P\hat{A}P^{-1}$ and $\underline{B}=P\int_0^h\exp(J_0\tau)P^{-1}\underline{B}_0 d\tau:=P\hat{B}$, respectively.

We introduce the following coordinate transformation:
\begin{align}\label{eqn_coordinate_transformation}
    \hat{x}_k=P^{-1}x_k
\end{align}
which changes the dynamics \eqref{eqn_discrete_system} to $\hat{x}_{k+1}=\hat{A}\hat{x}_k+\hat{B}\underline{u}_k$. Then, the system can be decoupled into two subsystems
    \begin{align*}
        \hat{x}_{1,k+1}=\hat{A}_1 \hat{x}_{1,k}+\hat{B}_1 \underline{u}_k\quad\hat{x}_{2,k+1}=\hat{x}_{2,k} + \hat{b} \underline{u}_k,
    \end{align*}    
where 
\begin{subequations}\label{eqn_sub_coordinate_transformation}
    \begin{align}
        \hat{x}_1&=[I_3,0_{3\times 1}]P^{-1}x:=P_1 x,\\
        \hat{x}_2&=[0_{1\times 3},1]P^{-1}x:=P_2x\\
        \hat{A}_1&=[I_3,0_{3\times 1}]\hat{A}[I_3,0_{3\times 1}]\T\\
        \hat{B}_1&=[I_3,0_{3\times 1}]\hat{B},~~\hat{b}=[0_{1\times 3},1]\hat{B}.
    \end{align}
\end{subequations}
Since all eigenvalues of $\hat{A}_1$ lie in the unit circle, there exist positive definite symmetric matrices $W,Q\in\mathbb{R}^{3\times 3}$, such that
\begin{align}\label{eqn_LMI}
    \hat{A}_1\T W\hat{A}_1-W=-Q.
\end{align}

Consider system \eqref{eqn_discrete_system}, the MPC strategy is formulated as follows: 
\begin{subequations}\label{eqn_MPClaw}
    \begin{align}
        \min_{\bu_k} ~&J(x_k,\bu_k)=\sum\nolimits_{i=0}^{N-1}l(x_{i\vert k},u_{i\vert k})+V_f(x_{N\vert k})\\
        \sth~&x_{0\vert k}=x_k,\\
        &x_{i+1\vert k}=\underline{A}x_{i\vert k}+\underline{B}u_{i\vert k},i\in\{0,\cdots,N-1\}\\
        &\vert u_{i\vert k}\vert \leq \Delta_{k+i},i\in\{0,\cdots,N-1\}.
    \end{align}
\end{subequations}
where $\bu_k=[u_{0\vert k},u_{1\vert k},\cdots,u_{N-1\vert k}]\T$ is the predicted future control inputs, $l:\mathbb{R}^4\times\mathbb{R}\to\mathbb{R}_{\geq 0}$ and $V_f:\mathbb{R}^4\to \mathbb{R}_{\geq 0}$ are the stage cost and terminal cost, respectively, given as follows:
\begin{subequations}\label{eqn_MPCcost}
\begin{align}
    l(x,u)&=\|x\|_{P_1\T P_1}^2+\Psi_\Delta(P_2 x)+u^2\label{eqn_cost}\\
    V_f(x)&=\Theta \|x\|_M^2\label{eqn_terminal}
\end{align}   
\end{subequations}
with $\Psi_\Delta(z)=\int_0^z \sigma_\Delta(s)ds$, and
\begin{subequations}\label{eqn_MPC_para}
    \begin{align}
        M&=(P^{-1})\T\begin{bmatrix}
        W & 0_{3\times 1} \\ 0_{1\times 3} & 1\\
        \end{bmatrix}P^{-1}\\
        \Theta&\geq \dfrac{1}{\min\{\frac{1}{2}\lambda_{\min}(Q),\vert \hat{b}\vert ,\hat{b}^2/\Gamma\}}\label{eqn_eqn_MPC_para_theta}\\
        \Gamma&=\hat{b}^2+\|\hat{B}_1\|_W^2+\frac{1}{\varepsilon}\|\hat{B}_1\|^2+\vert \hat{b}\vert\\
        \varepsilon&=\frac{\lambda_{\min}(Q)}{2\lambda_{\max}(\hat{A}_1\T WW\T \hat{A}_1)}.
    \end{align}
\end{subequations}
The solution to the optimization problem \eqref{eqn_MPClaw} is denoted by $\bu_k^*$. The first control input from the optimal sequence $\bu_k^*$ is applied to the system \eqref{eqn_discrete_system}, \ie
\begin{align}\label{eqn_MPCfeedback}
    \underline{u}_k=u^*_{0\vert k}.
\end{align}
Then, the following result is established:
\begin{thm}\label{thm_1}
    Consider the discrete-time subsystem \eqref{eqn_discrete_system} with the control input \eqref{eqn_MPCfeedback} generated by the MPC strategy \eqref{eqn_MPClaw}. Then, the origin $x_k=0$ of the resulting closed-loop system is UGAS.
\end{thm}
\begin{proof}
    See Appendix \ref{apx_thm1}.
\end{proof}

\begin{rem}
Since $\Psi_\Delta(\cdot)$ in the stage cost is defined as the integral of the saturation function $\sigma_\Delta(\cdot)$, satisfying $\mathcal{D} \sigma_\Delta(\cdot)>0$, it follows that $\Psi(\cdot)$ is convex. Hence, the stage cost remains convex, and the corresponding MPC formulation \eqref{eqn_MPClaw} is also convex.
\end{rem}

\subsection{Stability of the Continuous-Time Closed-loop System}
Although the previous subsection has demonstrated the UGAS for the proposed MPC strategy \eqref{eqn_MPClaw}, it is essential to analyze the behavior between sampling times to conclude UGAS for the continuous-time closed-loop system. The continuous-time closed-loop linear subsystem \eqref{eqn_separate_system} and \eqref{eqn_MPCfeedback} can be rewritten as
\begin{align}\label{eqn_lsys}
    \dot{x}=\tilde{f}_{out}(t,x):=\underline{A}_0 x(t)+\underline{B}_0 \underline{u}(t)
\end{align}
where $\underline{A}_0$ and $\underline{B}_0$ are given by \eqref{eqn_A0B0}, and the control input $\underline{u}(t)=u_\MPC (x(kh)),t\in[kh,kh+h)$ is generated by the MPC strategy \eqref{eqn_MPClaw}. Hence, for all $t\in[kh,kh+h]$, one has $x(t)=e^{\underline{A}_0(t-kh)}x(kh)+\int_0^{t-kh}e^{\underline{A}_0\tau}d\tau \underline{B}_0 \underline{u}(kh)$. Therefore, there exist constant $c_1>0$ and $c_2>0$, such that
\begin{align}\label{eqn_continuous_leq}
    &\|x(t)\|=\|e^{\underline{A}_0(t-kh)}x(kh)+\int_0^{t-kh}e^{\underline{A}_0\tau}d\tau \underline{B}_0 \underline{u}(kh)\|\notag\\
    &\quad \leq \|e^{\underline{A}_0(t-kh)}\|\|x(kh)\|+\|\int_0^{t-kh}e^{\underline{A}_0\tau} d\tau \underline{B}_0\|\vert \underline{u}(kh)\vert\notag\\
    &\quad \leq c_1\|x(kh)\|+c_2\vert \underline{u}(kh)\vert.
\end{align}
To conclude UGAS for the continuous-time closed-loop subsystem \eqref{eqn_lsys}, we have the following result:
\begin{lem}\label{lem_5}
    Consider the MPC feedback \eqref{eqn_MPCfeedback} generated by the MPC strategy \eqref{eqn_MPClaw}, there exists $c_3>0$, such that $|\underline{u}_k|\leq c_3\|x_k\|$ for any $k\in\mathbb{N}$.
\end{lem}
\begin{proof}
    See Appendix \ref{apx_prop1}
\end{proof}

With Lemma \ref{lem_5} established, one has the following result.
\begin{thm}\label{thm_2}
    The origin $x=0$ of the continuous-time closed-loop subsystem \eqref{eqn_lsys} is UGAS.
\end{thm}
\begin{proof}
    Based on Lemma \ref{lem_5}, inequality \eqref{eqn_continuous_leq} can be further rewritten as $\|x(t)\|\leq (c_1+c_2c_3)\|x(kh)\|,~~\forall t\in[kh,kh+h]$. Based on Theorem \ref{thm_1}, the origin of the discrete-time closed-loop subsystem \eqref{eqn_discrete_system}, \eqref{eqn_MPCfeedback} is UGAS, which implies there exists $\beta_0\in\mathcal{KL}$, such that $\|x_k\|\leq \beta_0(\|x_0\|,k)$ for any $x_0\in\mathbb{R}^4$ and $k\in\mathbb{N}$. Therefore, for any $x(0)\in\mathbb{R}^4$ and $t\geq 0$, one has
\begin{align}\label{eqn_kl_function_outer}
    \|x(t)\|&\leq (c_1+c_2c_3) \left\| x\left(\left \lfloor t/h\right\rfloor h\right)\right\|\notag\\
    &\leq (c_1+c_2 c_3)\beta_0(\|x(0)\|,\lfloor t/h\rfloor)\notag\\
    &:=\beta(\|x(0)\|,t)
\end{align}
where $\lfloor\cdot \rfloor$ is the floor function. 
Since $\beta_0\in\mathcal{KL}$, one has $\beta\in\mathcal{KL}$. By applying \cite[Lemma 4.5]{khalil2002nonlinear}, the origin $x=0$ is UGAS for the continuous-time closed-loop subsystem \eqref{eqn_lsys}.
\end{proof}
\begin{rem}
    Compared with \cite{izadi2024high}, the proposed MPC formulation \eqref{eqn_MPClaw} employs stage costs whose input terms have the same order as the state terms in the terminal cost. This structural property constitutes a sufficient condition for the validity of Lemma \ref{lem_5} and plays a key role in the proof of Theorem \ref{thm_2}.
\end{rem}

In view of \eqref{eqn_exCL} and \eqref{eqn_MPCfeedback}, the closed-loop system of the entire outer-loop system can be written as
\begin{align}\label{eqn_out_closed}
    \dot{x}_{out}=\tilde{F}_{out}(t,x_{out}):=Ax_{out}(t)+Bu(t)
\end{align}
where $x_{out}=(\tilde{p},\tilde{v},\mu_d,\eta)\in\mathbb{R}^{12}=\overline{\mathrm{vec}}\left([x^1,x^2,x^3]\T\right)$, the system matrices are given by
\begin{align*}
    A=\begin{bmatrix}
        0_{3\times3} & I_3 & 0_{3\times 3} & 0_{3\times 3}\\
        0_{3\times 3} & -D & I_3 & 0_{3\times 3} \\ 
        0_{3\times 3} & 0_{3\times 3} & -\frac{1}{\gamma}I_3 & \frac{\alpha}{\gamma}I_3\\
        0_{3\times 3} & 0_{3\times 3} & 0_{3\times 3} & -\frac{1}{\gamma}I_3
    \end{bmatrix},B=\begin{bmatrix}
        0_{3\times 3} \\ 0_{3\times 3} \\ 0_{3\times 3} \\ \frac{\alpha}{\gamma}I_3
    \end{bmatrix},
\end{align*}
and the control input
\begin{align}\label{eqn_outerloop_input}
    u(t)&=u_\MPC(x_{out}(kh)) :=\begin{bmatrix}
        u_{x\MPC}(x^1(kh)) \\ u_{y\MPC}(x^2(kh)) \\ u_{z\MPC}(x^3(kh))
    \end{bmatrix}
\end{align}
is generated by three distinct MPC controllers for all $t\in[kh,kh+h)$, which are formulated in \eqref{eqn_MPClaw}. Based on Theorem \ref{thm_2}, the following result is established.
\begin{thm}\label{thm_3}
    The origin $x_{out}=0$ is UGAS for the continuous-time closed-loop system \eqref{eqn_out_closed}.
\end{thm}

Theorem \ref{thm_3} is a direct consequence of Theorem \ref{thm_2} together with the definition of $x_{out}$. At this point, Problem \ref{prob_outer} is solved.
\subsection{Hybrid Feedback Design for Inner-Loop Subsystem}
From Lemma \ref{lem_gamma_alpha}, it follows that under constraint $\underline{u}_k\leq \Delta_k$ and initial constraint $\underline{\mu_d}(0)\leq \Delta_0$ and $\underline{\eta}(0)\leq \Delta_0$, the virtual input $\mu_d(t)$ and the virtual jerk $\eta(t)$ are bounded. Moreover, Assumption \ref{assum_1} ensures that the derivatives of $r(t)$ up to the order $4$ are bounded. According to Appendix \ref{apx_omega_d}, the expressions for $\omega_d$ and $\dot{\omega}_d$ are given in terms of the first to fourth derivatives of $r(t)$, $\mu_d(t),\eta(t)$, and $u(t)$, which implies $\omega_d(t)$ and $\dot{\omega}_d(t)$ are bounded. Consequently, the dynamics of the desired attitude can be written as
\begin{align}\label{eqn_desired_innerloop_kinematics}
    \left.\begin{aligned}
        \dot{R}_d&=R_d\omega_d^\times\\
        \dot{\omega}_d&\in m\mathbb{B}
    \end{aligned}\right\}(R_d,\omega_d)\in\mathcal{W}_d
\end{align}
where $m>0$ and $\mathcal{W}_d\subset \SO(3)\times\mathbb{R}^3$ is a compact subset. The attitude error and angular velocity error are defined as $\tilde{R}=R_d\T R, \tilde{\omega}=\omega-\tilde{R}\T \omega_d$. From \eqref{eqn_model_inner_omega} and \eqref{eqn_desired_innerloop_kinematics}, one obtains the following error dynamics:
    \begin{align}
        \dot{\tilde{R}}&=\tilde{R}\tilde{\omega}^\times,        I_b\dot{\tilde{\omega}}=\Sigma(\tilde{R},\tilde{\omega},\omega_d)\tilde{\omega}-\Upsilon(\tilde{R},\omega_d,\dot{\omega}_d)+\tau\label{eqn_inner_error}
    \end{align}
with the function $\Upsilon(\tilde{R},\omega_d,\dot{\omega}_d)=I_b\tilde{R}\T \dot{\omega}_d+(\tilde{R}\T\omega_d)^\times I_b\tilde{R}\T\omega_d$,$ \Sigma(\tilde{R},\tilde{\omega}$,$\omega_d)=(I_b\tilde{\omega})^\times+(I_b\tilde{R}\T\omega_d)^\times-((\tilde{R}\T \omega_d)^\times I_b+I_b(\tilde{R}\T\omega_d)^\times)$. 

Then, the hybrid controller proposed in \cite{Wang_Tayebi2022} is employed to ensure that the attitude $R$ converges semi-globally exponentially to $R_d$, where the set of parameters $\mathcal{P}_U=\{\Theta_h, A_{\theta}, u_\theta, \gamma_\theta,\delta\}$ is chosen according to \cite{Wang_Tayebi2022}:
\begin{align}\label{eqn_hybrid_input}
    \underbrace{
    \begin{aligned}
        \tau&=\Upsilon(\tilde{R},\omega_d,\dot{\omega}_d)-2k_R  e_R -k_\omega \tilde{\omega}\\
        \dot{\theta}&= -k_\theta \left(\gamma_\theta \theta +2u_\theta\T\mathcal{R}(\theta,u_\theta)^\top e_R\right)
    \end{aligned}}_{(\tilde{R},\theta)\in\mathcal{F}_\theta}~~
    \underbrace{
    \begin{aligned}
        \\
        \theta^+\in g_\theta(\tilde{R},\theta)
    \end{aligned}}_{(\tilde{R},\theta)\in\mathcal{J}_\theta}
\end{align}
where $e_R:= \mathcal{R}(\theta,u_\theta)\psi(A_\theta \tilde{R}\mathcal{R}(\theta,u_\theta))$, $k_R,k_\omega,k_\theta>0$, and  $g_\theta(\tilde{R},\theta):=\{\theta\in\Theta_h: \theta={\arg\min}_{\theta'\in\Theta_h}U(\tilde{R},\theta')\}$ 
with $U(\tilde{R},\theta) = \tr(A_\theta(I-\tilde{R}\mathcal{R}(\theta,u_\theta)))+\frac{\gamma_\theta}{2}\theta^2$. The flow and jump sets are defined as $\mathcal{F}_\theta:=\{(\tilde{R},\theta)\in\SO(3)\times\mathbb{R}:\mu_U(\tilde{R},\theta)\leq \delta\},\mathcal{J}_\theta:=\{(\tilde{R},\theta)\in\SO(3)\times \mathbb{R}:\mu_U(\tilde{R},\theta)\geq \delta\}$ with $\mu_U:=U(\tilde{R},\theta)-\min_{\theta'\in\Theta_h} U(\tilde{R},\theta')$. Define the new state $x_{in}:=(\tilde{R},\theta,\tilde{\omega},R_d,\omega_d)\in\mathcal{X}_{in}:=\SO(3)\times\mathbb{R}\times\mathbb{R}^3\times\mathcal{W}_d$. In view of \eqref{eqn_inner_error}, \eqref{eqn_hybrid_input}, one has the following hybrid closed-loop system:
\begin{align}\label{eqn_hybrid_inner_closed_loop}
    \left\{
    \begin{aligned}
        &\dot{x}_{in}\in F_{in}(x_{in}),x_{in}\in\mathcal{F}_{in}:=\{x_{in}\in\mathcal{X}_{in}:(\tilde{R},\theta)\in\mathcal{F}_\theta\}\\
        &x^+_{in}\in G_{in}(x_{in}),x_{in}\in \mathcal{J}_{in}:=\{x_{in}\in\mathcal{X}_{in}:(\tilde{R},\theta)\in\mathcal{J}_\theta\}
    \end{aligned}
    \right.
\end{align}
where the flow and jump maps are given by 
    \begin{align*}
        F_{in}(x_{in})&:=\begin{pmatrix}
            \tilde{R}\tilde{\omega}^\times \\
            f_\theta(\tilde{R},\theta)\\
            I_b^{-1}(\Sigma(\tilde{R},\tilde{\omega},\omega_d)\tilde{\omega}-2k_Re_R-k_\omega\tilde{\omega})\\
            R_d\omega_d^\times\\
            m\mathbb{B}
        \end{pmatrix}\\
        G_{in}(x_{in})&:=\left(\tilde{R},g_\theta(\tilde{R},\theta),\tilde{\omega},R_d,\omega_d\right).
    \end{align*}
The stability of the closed-loop system \eqref{eqn_hybrid_inner_closed_loop} follows from \cite[Proposition 1]{Wang_Tayebi2022} :
\begin{thm}[\cite{Wang_Tayebi2022}] \label{thm_4}
    Let $k_R,k_\omega,k_\theta>0$. Then, the closed set $\mathcal{A}_{in}=\{x_{in}\in\mathcal{X}_{in}:\tilde{R}=I_3,\theta=0,\tilde{\omega}=0\}$ is SGES and the number of jumps is finite for the hybrid closed-loop system \eqref{eqn_hybrid_inner_closed_loop}. More precisely, there exist $k,\lambda>0$ such that each maximal solution $x_{in}$ to the hybrid system \eqref{eqn_hybrid_inner_closed_loop} satisfies $\vert x_{in}(t,j)\vert_{\mathcal{A}_{in}} ^2\leq k\exp(-\lambda(t+j))\vert x(0,0)\vert _{\mathcal{A}_{in}}^2$ for all $(t,j)\in\mathrm{dom}~x_{in}$ and every initial condition $x_{in}(0,0)\in\Omega_c\times \mathcal{W}_d$ with a compact set $\Omega_c\subset \SO(3)\times\mathbb{R}\times\mathbb{R}^3$. 
\end{thm}

Since the number of jumps is finite, the SGES can be viewed as exponential stability in the classical sense (exponential convergence over time), \ie, for every compact set $\Omega_c\subset \SO(3)\times\mathbb{R}\times\mathbb{R}^3$ and every initial condition $x_{in}(0)\in\Omega_c\times\mathcal{W}_d$, there exist $k',\lambda'>0$ such that
\begin{align}\label{eqn_in_exp}
    \vert x_{in}(t)\vert _{\mathcal{A}_{in}}\leq k' \exp(-\lambda't)\vert x_{in}(0)\vert _{\mathcal{A}_{in}}\quad\forall t\geq 0.
\end{align}

\section{Stability of the Overall Cascaded System}
In view of \eqref{eqn_model}, \eqref{eqn_thrust}, \eqref{eqn_exCL}, \eqref{eqn_outerloop_input}, and \eqref{eqn_hybrid_input}, the overall closed-loop system can be written as the following cascaded form:
\begin{subequations}\label{eqn_cascaded}
    \begin{align}
        &\dot{x}_{out}=F_{out}(t,x_{out},x_{in})\label{eqn_cascaded_out}\\
        &\underbrace{\dot{x}_{in}=F_{in}(x_{in})
        \quad\quad\quad~~}_{{x_{in}\in\mathcal{F}_{in}}}\underbrace{
            x_{in}^+\in G_{in}(x_{in})
        }_{x_{in}\in\mathcal{J}_{in}}\label{eqn_cascaded_in}
    \end{align}     
\end{subequations}
where $F_{in}$ and $G_{in}$ are given by \eqref{eqn_hybrid_inner_closed_loop}, and  $F_{out}(t,x_{out},x_{in}) := \tilde{F}_{out}(t,x_{out})+T\Lambda R_d\left(I_3-\Phi(x_{in})\right)e_3$ with $\Lambda=[0_{3\times 3},I_3,0_{6\times 3}\T]\T$, $\tilde{F}_{out}$ given in \eqref{eqn_out_closed}, and $\Phi(\cdot):\mathcal{X}_{in}\to\SO(3)$ denoting the projection map $\Phi(x_{in})=\tilde{R},\forall x_{in}=(\tilde{R},\theta,\tilde{\omega},R_d,\omega_d)\in\mathcal{X}_{in}$.

\begin{thm}
\label{thm_final}
Consider the cascaded closed-loop system \eqref{eqn_cascaded}. Define the closed set $\mathcal{A}:=\{(x_{out},x_{in})\in\mathbb{R}^{12}\times\mathcal{X}_{in}:x_{out}=0,\;x_{in}\in\mathcal{A}_{in}\}$ with $\mathcal{A}_{in}$ defined in Theorem \ref{thm_4}. Then, the set $\mathcal{A}$ is USGAS for the system \eqref{eqn_cascaded}.
\end{thm}

\begin{proof}
    The proof proceeds by first establishing the boundedness of the solutions to \eqref{eqn_cascaded}. Note that the outer-loop subsystem \eqref{eqn_cascaded_out} remains decoupled along the $x$-, $y$-, and $z$-axes. In view of \eqref{eqn_lsys} and \eqref{eqn_cascaded_out}, for each axis, one has
        \begin{align}\label{eqn_cascaded_closed_outerloop}
            \dot{x}=\underline{A}_0 x+\underline{B}_0 u_\MPC(x_k)+\Lambda_0d(t), \quad t\in[t_k,t_{k+1})
        \end{align}
    where $\underline{A}_0$ and $\underline{B}_0$ are given by \eqref{eqn_A0B0}, $t_k=kh,\Lambda_0=[0,1,0,0]\T$, $d(t):=\Pi\left(TR_d(I_3-\Phi(x_{in}))e_3\right)$, and $\Pi:\mathbb{R}^3\to\mathbb{R}$ denotes the canonical projection onto the $i$-th component (with $i$ omitted). According to Theorem \ref{thm_4}, there exists $c_d,\lambda_d>0$ such that, for every $x_{in}(0,0)\in\Omega_c\times\mathcal{W}_d$, $\vert d(t)\vert \leq c_d e^{-\lambda_d t}\vert x_{in}(0,0)\vert _{\mathcal{A}_{in}}$. Integrating \eqref{eqn_cascaded_closed_outerloop} over each sampling interval yields
    \begin{align}\label{eqn_closedloop_out_i}
        x_{k+1}=\underline{A}x_k+\underline{B}u_{\MPC}(x_k)+\omega_k
    \end{align}
    where $\omega_k=\int_0^h e^{\underline{A}_0(h-s)}\Lambda_0d(kh+s)ds$. From the expressions of $\omega_k$ and $d(t)$, one has $\|\omega_k\|\leq c_\omega e^{-\lambda_d kh}\vert x_{in}(0,0)\vert _{\mathcal{A}_{in}}$ for some $c_\omega>0$. Hence,
    \begin{align}\label{eqn_bounded_condition_1}
        \sum_{k=0}^\infty\|\omega_k\|<\infty,\quad \sum_{k=0}^\infty\|\omega_k\|^2<\infty
    \end{align}

    To establish the boundedness of \eqref{eqn_closedloop_out_i}, consider the optimal MPC cost $V_k(x):=J_k^*(x)$. From the definition of the cost function and the stability established in Appendix \ref{apx_thm1}, one has 
    \begin{subequations}
        \begin{align}
            \alpha_1(\|x\|)\leq V_k(x)&\leq \Theta\lambda_{\max{}}(M)\|x\|^2,\label{eqn_bounded_condition_2}\\
            V(F_k(x))-V_k(x)&\leq 0
        \end{align}
    \end{subequations}
    where $\alpha_1\in\mathcal{K}_{\infty},F_k(x):=\underline{A}x+\underline{B}u_\MPC(x)$. We next consider the perturbation induced by $\omega_k$. Let $z=F_k(x)$ and denote by $\mathbf{u}^*=\{u_0^*,\cdots,u_{N-1}^*\}$ the optimal input sequence associated with $z$. Since the MPC problem only involves input constraints, $\mathbf{u}^*$ remains feasible when the initial state is perturbed to $z+\omega_k$. By the optimality of $V_{k+1}(z+\omega_k)$, one has
    \begin{align*}
        &V_{k+1}(F_k(x)+\omega_k)-V_k(x)\leq V_{k+1}(z+\omega_k)-V_{k+1}(z)\\
        \leq &J_{k+1}(z+\omega_k,\mathbf{u}^*)-J_{k+1}(z,\mathbf{u}^*)
    \end{align*}
    Let $z_i$ and $\tilde{z}_i$ denote the predicted states generated by $\mathbf{u}^*$ from the initial conditions $z$ and $z+\omega_k$, respectively. Since the prediction model is linear, $\tilde{z}_i-z_i=\underline{A}^i\omega_k, i=0,\cdots,N$. Consequently,
    \begin{align*}
        &J_{k+1}(z+\omega_k,\mathbf{u}^*)-J_{k+1}(z,\mathbf{u}^*)\\
        =&\sum_{i=0}^{N-1}\left(\mathcal{L}(\tilde{z}_i)- \mathcal{L}(z_i)\right) + V_f(\tilde{z}_N)-V_f(z_N).
    \end{align*}
    where $\mathcal{L}(x)=\|P_1x\|^2+\Psi_\Delta(P_2x)$. From the definition of $\Psi_\Delta$ that $\Psi_\Delta(x)=\int_0^x\sigma_\Delta(s)ds$, and the property of $\sigma_\Delta$ that $\mathcal{D}\sigma_\Delta(x)\in(0,1]$, one can show that there exist constants $c_4,c_5,c_6,c_7,c_8>0$ such that
    \begin{align*}
        \mathcal{L}(z_i+\underline{A}^i\omega_k)-\mathcal{L}(z_i)&\leq c_4\|z_i\|\|\omega_k\|+c_5\|\omega_k\|+c_6\|\omega_k\|^2\\
        V_f(\tilde{z}_N)-V_f(z_N)&\leq c_7\|z_N\|\|\omega_k\|+c_8\|\omega_k\|^2
    \end{align*}
    Moreover, for any $i\in[0,N]$, there exist $\tilde{c}_i$, such that $\|z_i\|\leq \tilde{c}_i\sqrt{V_{k+1}(z)}\leq \tilde{c}_i\sqrt{V_k(x)}$. Consequently, there exist $\tilde{\gamma}_1,\tilde{\gamma}_2,\tilde{\gamma}_3>0$, such that
    \begin{align}
        &V_{k+1}(F_k(x)+\omega_k)-V_k(x)\notag\\
        \leq & \tilde{\gamma}_1 \sqrt{V_k(x)}\|\omega_k\|+\tilde{\gamma}_2\|\omega_k\|+\tilde{\gamma}_3\|\omega_k\|^2.\label{eqn_bounded_condition_3}
    \end{align}
    Therefore, \eqref{eqn_bounded_condition_1}, \eqref{eqn_bounded_condition_2}, and \eqref{eqn_bounded_condition_3} satisfy the conditions given in \cite[Proposition 3]{nesic2004uniform}, which implies the solutions of \eqref{eqn_closedloop_out_i} are uniformly bounded.

    For $t\in[t_k,t_{k+1})$, there exist $\bar{c}_1,\bar{c}_2,\bar{c}_3,\bar{c}_4>0$ such that
    \begin{align*}
        \|x(t)\|&\leq \left\|e^{A_0(t-t_k)}\right\|\|x_k\|+\int_{0}^{t-t_k}\left\|e^{\underline{A}_0(t-t_k-s)}\underline{B}_0\right\|\vert u_k\vert ds \\
        &~~~~+\int_0^{t-t_k}\left\|e^{\underline{A}_0(t-t_k-s)}\Lambda_0\right\|\vert d(t_k+s)\vert ds\\
        &\leq \bar{c}_1\|x_k\|+\bar{c}_2\vert u_k\vert +\bar{c}_3 \sup_{s\in[t_k,t_{k+1})}\vert d(s)\vert \\
        &\leq \bar{c}_4\|x_k\|+\bar{c}_3\sup_{s\in[t_k,t_{k+1})}\vert d(s)\vert
    \end{align*}
    where we made use of the fact $\vert u(t)\vert \leq c_3\|x_k\|$ for some $c_3>0$ and all $t\in[t_k,t_{k+1})$, which is established in Lemma \ref{lem_5}. 

    Since $\|x_k\|$ is uniformly bounded and $d(t)$ decays exponentially, it follows that the solution of the outer-loop system \eqref{eqn_cascaded_out} is bounded. Together with the boundedness of the solution of the inner-loop subsystem \eqref{eqn_cascaded_in} established in Theorem \ref{thm_4}, the solution of the overall cascaded system \eqref{eqn_cascaded} is bounded. By Theorem \ref{thm_3}, the equilibrium $x_{out}=0$ of the outer-loop subsystem \eqref{eqn_out_closed} is UGAS, and by Theorem \ref{thm_4}, the set $\mathcal{A}_{in}$ is SGES for the inner-loop subsystem \eqref{eqn_cascaded_in}. By applying \cite[Theorem 4]{andrien2024model}, the set $\mathcal{A}$ is USGAS for the cascaded system \eqref{eqn_cascaded}.
\end{proof}

\begin{rem}
The semi-global property of the overall closed-loop system stems from the fact that, in our proof, the initial angular velocity $\omega(0)$ is required to belong to an arbitrary compact subset of $\mathbb{R}^3$, with no restrictions on the initial attitude, \textit{i.e.,} $R(0) \in \SO(3)$. 
\end{rem}

\section{Simulation Results}
In this section, simulation results are presented to illustrate the performance of the proposed control framework composed of the MPC-based outer-loop controller and the hybrid geometric inner-loop controller. The MPC formulation is solved online in MATLAB using CasADi \cite{andersson2018casadi}.
The proposed MPC strategy \eqref{eqn_MPClaw} and hybrid controller \eqref{eqn_hybrid_input} in Section \ref{sec_III} are referred to as ``MPC+Hybrid''. For comparison, we also consider the same MPC strategy \eqref{eqn_MPClaw} combined with the following classical non-hybrid controller, referred to as "MPC+Non-Hybrid": $\tau = \Upsilon(\tilde{R},\omega_d,\dot{\omega}_d)-2k_R\psi(A\tilde{R}) - k_\omega \tilde{\omega}$, which is modified from \eqref{eqn_hybrid_input} by taking $\theta \equiv 0$. 

The system parameters are selected as follows: $I_b=\mathrm{diag}(0.0159,0.0150,0.0297)$, $g=9.81$ and $T_{\max}=25$. The reference trajectory is given by $r(t) = [3\sin(2t), 3\cos(2t), 8+4\cos(t)]\T$ and $\psi_d(t) = 0.5t$.
For the outer-loop controller, the sampling period of the MPC is selected as $h=0.05$, and the prediction horizon is chosen as $N=25$. For the parameters $\gamma$ and $\alpha$ in \eqref{eqn_exCL}, $\gamma$ is chosen as $0.9$ times its upper bound, and $\alpha$ is chosen at its upper bound. The function $\Psi_\Delta(\cdot)$ in the stage cost \eqref{eqn_cost} is chosen as $\Psi_\Delta(x)= \ln\cosh(x)$, which is the integral of the saturation function $\sigma_\Delta(x)=\tanh(x)$. The positive definite matrix $Q$ in \eqref{eqn_LMI} is chosen as $Q = I_3$. The parameter $\Theta$ in \eqref{eqn_terminal} is chosen as $1.1$ times the corresponding lower bound required in \eqref{eqn_eqn_MPC_para_theta}. For the inner-loop controller, the control gains are selected as $k_R=1.5$, $k_\omega=0.2$, and $k_\theta=10$. The parameters in hybrid mechanism are chosen as $A_\theta=\mathrm{diag}(2,4,6),\Theta_h=\{0.9\pi\},\gamma_\theta = \frac{7}{\pi^2}, \delta = 0.324$, and $u_\theta =[0,0.6325,0.7746]\T$. The initial conditions are chosen as $p(0)=0,v(0)=0,\mu_d(0)=0,\eta(0)=0,\theta(0,0)=0,\omega(0)=0$ and $R(0)=R_d(0)\mathcal{R}(\pi,e_1)$.

\begin{figure}[!h]
    \centering
    \includegraphics[width=0.75\linewidth]{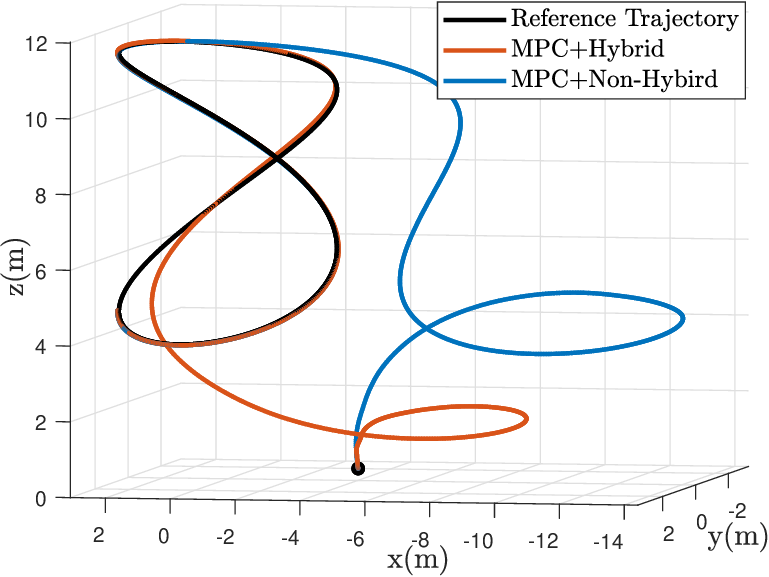}
    \vspace{-3mm}
    \caption{3D plot of the reference and actual trajectories with the black dot indicating the starting point of the actual trajectory.}
    \vspace{-2mm}
    \label{fig_1}
\end{figure}
\begin{figure}[!t]
    \centering
    \includegraphics[width=0.75\linewidth]{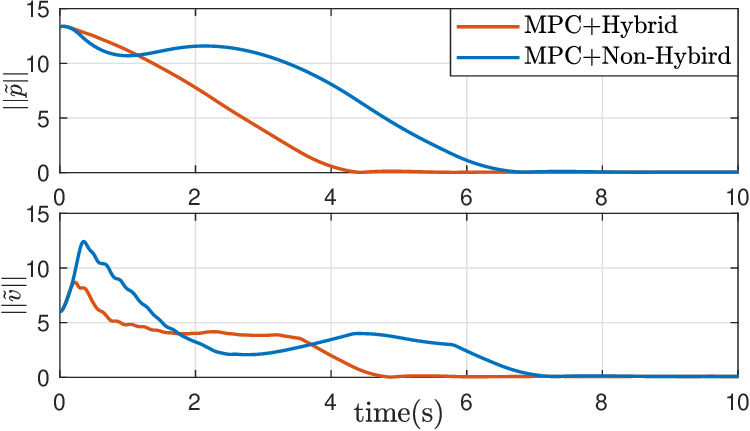}
    \vspace{-3mm}
    \caption{The time evolution of the tracking errors $\tilde{p}$ and $\tilde{v}$.}
    \vspace{-2mm}
    \label{fig_2}
\end{figure}
\begin{figure}[!t]
    \centering
    \includegraphics[width=0.75\linewidth]{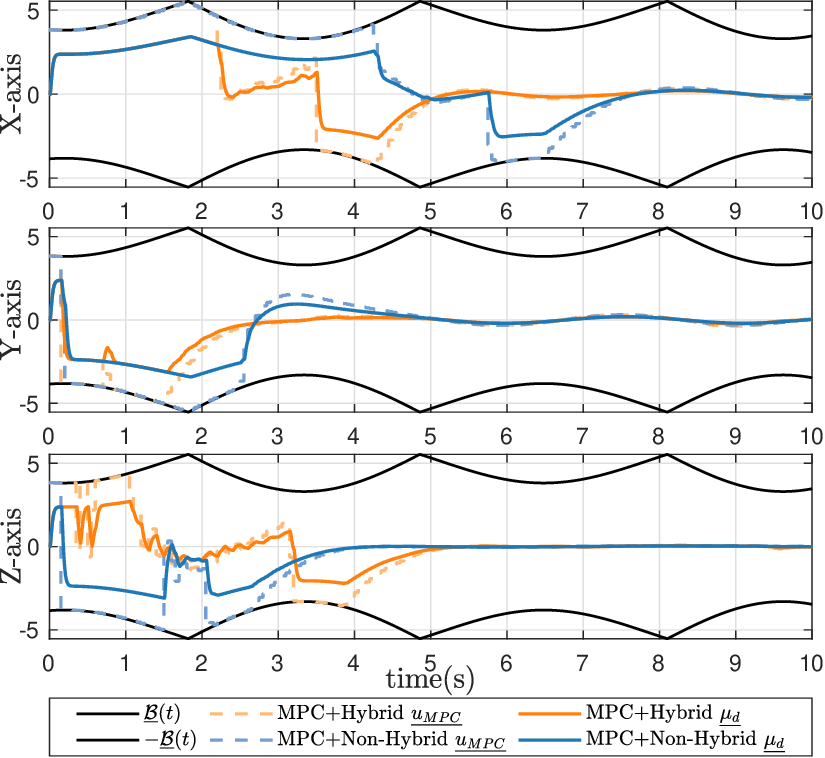}
    \vspace{-3mm}
    \caption{The time evolution of desired input $\mu_d$ and MPC input $u_{\MPC}$.}
    \vspace{-5mm}
    \label{fig_3}
\end{figure}
The reference and resulting position trajectories are shown in Fig.~\ref{fig_1}, which shows that the proposed MPC+Hybrid exhibits better tracking performance than the MPC+Non-Hybrid. A similar result can also be observed in Fig.~\ref{fig_2}. Fig.~\ref{fig_3} illustrates the virtual input $\mu_d$ and the MPC input $u_{\mathrm{MPC}}$ along the $x$-, $y$-, and $z$-axes under two different control strategies, which shows that the desired input $\mu_d$ always satisfies the constraints. Meanwhile, due to the scaling factor $\alpha$, the virtual input $\mu_d$ maintains a certain margin from the constraint boundary.

\section{Conclusions}
This paper addresses the trajectory tracking problem of quadrotor UAVs and develops an MPC-based hierarchical control framework. The designed outer-loop MPC guarantees UGAS of the outer-loop system in a continuous-time sense. Moreover, the introduction of the scale factor allows the constraints on the virtual input to be directly applied to the augmented virtual input. Furthermore, we prove that the cascaded overall closed-loop system consisting of the outer-loop controller and the hybrid inner-loop controller is USGAS.  Simulation results show that the proposed framework achieves satisfactory trajectory tracking performance without imposing restrictions on the initial attitude, and yields better transient responses than its counterpart with a non-hybrid inner-loop controller.

\appendix
\subsection{Proof of Lemma 1} \label{adx_lem1}
    Proof of item 1.    
    Since $0<\sigma_\Delta'(x)\leq 1$ for all $x\in\mathbb{R}$, it follows that $ax-\sigma_\Delta(ax)\geq 0$ when $x\geq 0$ and $ax-\sigma_\Delta(ax)\leq 0$ when $x\leq 0$. Consequently, $(ax-\sigma_\Delta(ax))f(x)\geq 0$ for all $x\in\mathbb{R}$, which yields $\sigma_\Delta(ax)f(x)\leq axf(x)$.

    Proof of item 2. Define $\rho_\Delta(x)=\sigma_\Delta(bx)-b\sigma_\Delta(x)$. Then, the  derivative $\rho_\Delta '(x)=b\left(\sigma_\Delta'(bx)-\sigma_\Delta '(x)\right)$ and $\rho_\Delta(0)=0$. Given $b\in[0,1]$ and the monotonicity of $\sigma_\Delta'$ (nonincreasing over $(0,+\infty)$, nondecreasing over $(-\infty,0)$), we have $\rho_\Delta'(x)\geq 0$ for all $x\in\mathbb{R}$. Hence, $\sigma_\Delta(bx)\geq b\sigma_\Delta(x)$ for $x\geq 0$ and $\sigma_\Delta(bx)\leq b\sigma_\Delta(x)$ for $x\leq 0$, implying $x\sigma_\Delta(bx)\geq bx\sigma_\Delta(x)$ for all $x\in\mathbb{R}$.

    Proof of item 3. By the Mean Value Theorem for integrals, for all $x\in\mathbb{R}$, $\Psi_\Delta(x)=\int_0^x\sigma_\Delta(s)ds=x\int_0^1\sigma_\Delta(tx)dt\leq x\sigma_\Delta(x)$.

\subsection{The expressions of $\omega_d$ and $\dot{\omega}_d$}\label{apx_omega_d}
From the expression of the desired attitude \eqref{eqn_desired_attitude}, the desired angular velocity $\omega_d$ and angular acceleration $\dot{\omega}_d$ are derived by first computing $\dot{R}_d$ and $\ddot{R}_d$, which are given by
\begin{subequations}\label{eqn_dot_ddot_Rd}
    \begin{align}
        \dot{R}_d &= [\dot{\mathbf{x}}_d,\dot{\mathbf{z}}_d \times \mathbf{x}_d + \mathbf{z}_d\times \dot{\mathbf{x}}_d, \dot{\mathbf{z}}_d]\\
        \ddot{R}_d & = [\ddot{\mathbf{x}}_d,\ddot{\mathbf{z}}_d \times \mathbf{x}_d + 2 \dot{\mathbf{z}}_d\times\dot{\mathbf{x}}_d + \mathbf{z}_d\times\ddot{\mathbf{x}}_d, \ddot{\mathbf{z}}_d].
    \end{align}    
\end{subequations}

The expressions of $\dot{\mathbf{x}}_d$, $\dot{\mathbf{z}}_d$, 
$\ddot{\mathbf{x}}_d$, and $\ddot{\mathbf{z}}_d$ are derived as follows. In view of \eqref{eqn_zd}, let $\nu_z = ge_3 - D\dot{r} - \ddot{r} - \mu_d$, so that $\mathbf{z}_d = \nu_z / \|\nu_z\|$. Its first- and second-order time derivatives are given by
\begin{subequations}\label{eqn_dot_ddot_zd}
    \begin{align}
        \dot{\mathbf{z}}_d&=\frac{\|\nu_z\|^2 \dot{\nu}_z -\nu_z\nu_z\T \dot{\nu}_z}{\|\nu_z\|^3}=\frac{(I_3 -\mathbf{z}_d\mathbf{z}_d\T)\dot{\nu}_z}{\|\nu_z\|} = \frac{ \pi_z\dot{\nu}_z}{\|\nu_z\|}\\
        \ddot{\mathbf{z}}_d &= \frac{1}{\|\nu_z\|}\left(\dot{\pi}_z\dot{\nu}_z + \pi_z\ddot{\nu}_z-\mathbf{z}_d\T \dot{\nu}_z\dot{\mathbf{z}}_d\right)
    \end{align}
\end{subequations}
where 
\begin{subequations}\label{eqn_dot_ddot_nuz}
    \begin{align}
        \dot{\pi}_z&=-\dot{\mathbf{z}}_d\mathbf{z}_d\T -\mathbf{z}_d\dot{\mathbf{z}}_d\T\\
        \dot{\nu}_z
        &= -D\ddot{r} - r^{(3)} + \frac{1}{\gamma} \mu_d -\frac{\alpha}{\gamma}\eta \\
        \ddot{\nu}_z &= -D r^{(3)} - r^{(4)} - \frac{1}{\gamma^2} \mu_d + \frac{2\alpha}{\gamma^2}\eta -\frac{\alpha^2}{\gamma^2}u.
    \end{align}   
\end{subequations}

From \eqref{eqn_xd}, let $\nu_x = \pi_z \mathbf{x}_b$, so that $\mathbf{x}_d = \nu_x / \|\nu_x\|$. Similarly to $\mathbf{z}_d$, its first- and second-order time derivatives are given by
\begin{subequations}
    \begin{align}
        \dot{\mathbf{x}}_d & = \frac{\pi_x \dot{\nu}_x}{\|\nu_x\|},\\
        \ddot{\mathbf{x}}_d & = \frac{1}{\|\nu_x\|}\left(\dot{\pi}_x \dot{\nu}_x+ \pi_x \ddot{\nu}_x -\mathbf{x}_d\T \dot{\nu}_x \dot{\mathbf{x}}_d\right)
    \end{align}
\end{subequations}
where 
\begin{subequations}
    \begin{align}
        \pi_x &= I_3 - \mathbf{x}_d\mathbf{x}_d\T \\
        \dot{\pi}_x & = -\dot{\mathbf{x}}_d \mathbf{x}_d\T - \mathbf{x}_d \dot{\mathbf{x}}_d\T\\
        \dot{\nu}_x &= \dot{\pi}_z \mathbf{x}_b +\pi_z \dot{\mathbf{x}}_b\\
        \ddot{\nu}_x & =\ddot{\pi}_z \mathbf{x}_b + 2 \dot{\pi}_z \dot{\mathbf{x}}_b + \pi_z \ddot{\mathbf{x}}_b  \\
        \ddot{\pi}_z& = - \ddot{\mathbf{z}}_d \mathbf{z}_d\T - 2 \dot{\mathbf{z}}_d\dot{\mathbf{z}}_d\T - \mathbf{z}_d \ddot{\mathbf{z}}_d\T .
    \end{align}
\end{subequations}

In view of \eqref{eqn_xb}, the first- and second-order time derivatives are given by
\begin{subequations}\label{eqn_dot_ddot_xb}
    \begin{align}
        \quad\dot{\mathbf{x}}_b &=\dot{\psi}_d\begin{bmatrix}
            -\sin\psi_d \\ \cos\psi_d \\ 0
        \end{bmatrix}\\
        \ddot{\mathbf{x}}_b &= \ddot{\psi}_d \begin{bmatrix}
            -\sin\psi_d \\ \cos\psi_d \\ 0
        \end{bmatrix} + (\dot{\psi}_d)^2 \begin{bmatrix}
            -\cos\psi_d \\ -\sin\psi_d \\ 0
        \end{bmatrix}.
    \end{align}   
\end{subequations}

Finally, from the kinematics of $R_d$, \ie, $\dot{R}_d=R_d\omega_d^\times$, the desired angular velocity $\omega_d$ and angular acceleration $\dot{\omega}_d$ are given by
\begin{subequations}
    \begin{align}
        \omega_d & =\mathrm{vec}\left(R_d\T \dot{R}_d\right)\\
        \dot{\omega}_d & =\mathrm{vec}\left(\dot{R}_d\T \dot{R}_d + R_d\T \ddot{R}_d\right)
    \end{align}
\end{subequations}
where $R_d$ is given by \eqref{eqn_desired_attitude}, $\dot{R}_d$ and $\ddot{R}_d$ are given by \eqref{eqn_thrust}-\eqref{eqn_xd} and \eqref{eqn_dot_ddot_Rd}-\eqref{eqn_dot_ddot_xb}.

\subsection{Proof of Lemma 3}\label{apx_lem3}
According to Assumption \ref{assum_1}, the first and second derivatives $\dot{r}(t)$ and $\ddot{r}(t)$ are Lipschitz continuous, that is, there exist known constants $L_1=\sup_{t\geq 0}\|\ddot{r}(t)\|$ and $L_2=\sup_{t\geq 0}\|r^{(3)}(t)\|$, such that
\begin{align*}\label{eqn_Lipschitz}
    \| \dot{r}(t_1)-\dot{r}(t)\| \leq L_1\vert t_1-t_2\vert \\
    \| \ddot{r}(t_1)-\ddot{r}(t_2)\| \leq L_2\vert t_1-t_2\vert 
\end{align*}
for every $t_1,t_2>0$. Consider $\mathcal{B}_1(\dot{r},\ddot{r})$, one has
\begin{align*}
    &\vert \mathcal{B}_1(\dot{r}(t_1),\ddot{r}(t_1))-\mathcal{B}_1(\dot{r}(t_2),\ddot{r}(t_2))\vert \\
    =&\left\vert e_3\T\left(-D\dot{r}(t_1)-\ddot{r}(t_1)+D\dot{r}(t_2)+\ddot{r}(t_2)\right)\right\vert\\
    \leq&\|\ddot{r}(t_1)-\ddot{r}(t_2)\|+\|D\|_2\cdot\|\dot{r}(t_1)-\dot{r}(t_2)\|\\
    \leq&\sqrt{3}\bar{L}\vert t_1-t_2\vert
\end{align*}
with $\bar{L}=\left(\bar{d}L_1+L_2\right)/\sqrt{3}$, $\bar{d} = \max \{d_1,d_2,d_3\}$,
    which implies the constraint $\mathcal{B}_1(\dot{r},\ddot{r})$ is Lipschitz continuous for time $t$. Consider $\mathcal{B}_2(\dot{r},\ddot{r})$, one has
\begin{align*}
    &\vert \mathcal{B}_2(\dot{r}(t_1),\ddot{r}(t_1))-\mathcal{B}_2(\dot{r}(t_2),\ddot{r}(t_2))\vert \\
    =&\left\vert \|ge_3-\ddot{r}(t_1)-D\dot{r}(t_1)\|-\|ge_3-\ddot{r}(t_2)-D\dot{r}(t_2)\|\right\vert\\
    \leq& \| -\ddot{r}(t_1) + \ddot{r}_2(t_2) - D(\dot{r}({t}_1) - \dot{r}(t_2))\|\\
    \leq&\|\ddot{r}(t_1)-\ddot{r}(t_2)\|+\|D\|_2\cdot\|\dot{r}(t_1)-\dot{r}(t_2)\|\\
    \leq&\sqrt{3}\bar{L}\vert t_1-t_2\vert
\end{align*}
Since the minimum operator is non-expansive under the infinity norm, we have
\begin{align*}
    &\vert \mathcal{B}(\dot{r}(t_1),\ddot{r}(t_1))-\mathcal{B}(\dot{r}(t_2),\ddot{r}(t_2))\vert\\
    =&\vert \min\{\mathcal{B}_1(\ddot{r}(t_1)),\mathcal{B}_2(\ddot{r}(t_1))\} - \min\{\mathcal{B}_1(\ddot{r}(t_2)),\mathcal{B}_2(\ddot{r}(t_2))\}\vert\\
    \leq&\max\{\vert \mathcal{B}_1(\ddot{r}(t_1))-\mathcal{B}_1(\ddot{r}(t_2))\vert,\vert \mathcal{B}_2(\ddot{r}(t_1))-\mathcal{B}_2(\ddot{r}(t_2))\vert \}\\
    \leq& \sqrt{3}\bar{L}\vert t_1-t_2\vert
\end{align*}

In view of \eqref{eqn_constraint}, one has
\begin{align*}
    \vert \underline{\mathcal{B}}(\dot{r}(t_1),\ddot{r}(t_1))-\underline{\mathcal{B}}(\dot{r}(t_2),\ddot{r}(t_2))\vert\leq \bar{L}\vert t_1-t_2\vert .
\end{align*}
This completes the proof.

\subsection{Proof of Lemma 4}\label{apx_lem4}
From $0<\gamma<\Delta/\bar{L}$, we have
\begin{align*}
    e^{-h/\gamma}<e^{-\bar{L}h/\Delta}<e^{-\ln (1+\bar{L}h/\Delta)}=\beta
\end{align*}
Therefore, $\frac{\beta-e^{-h/\gamma}}{1-e^{-h/\gamma}}>0$, which implies that the scale factor $\alpha$ is feasible. In view of \eqref{eqn_separate_system_4}, the evolution of the jerk of the system \eqref{eqn_separate_system} during one sampling period can be given by
\begin{align*}
    \underline{\eta}(t)&=e^{-(t-t_k)/\gamma}\underline{\eta}(t_k)+\left(1-e^{-(t-t_k)/\gamma}\right)\alpha \underline{u}_k
\end{align*}
for $t\in[t_k,t_{k+1}]$. On one hand, by the fact that $0< e^{-(t-t_k)/\gamma}\leq 1$ and $\alpha<1$, we have
\begin{align*}
    \vert \underline{\eta}(t)\vert&\leq e^{-(t-t_k)/\gamma}\vert \underline{\eta}(t_k)\vert +\left(1-e^{-(t-t_k)/\gamma}\right)\vert \underline{u}_k\vert\\
    &\leq e^{-(t-t_k)/\gamma}\Delta_k +\left(1-e^{-(t-t_k)/\gamma}\right)\Delta_k\\
    &=\Delta_k.
\end{align*}
On the other hand, we have
\begin{align*}
    \frac{\Delta_{k+1}}{\Delta_k}&\geq \frac{\Delta_{k+1}}{\Delta_{k+1}+\bar{L}h}\geq \frac{\Delta}{\Delta +\bar{L}h}=\beta.
\end{align*}
Therefore,
\begin{align*}
    \vert \underline{\eta}(t_{k+1})\vert &=\left\vert e^{-h/\gamma}\underline{\eta}(t_k)+\left(1-e^{-h/\gamma}\right)\alpha \underline{u}_k\right\vert\\
    &\leq e^{-h/\gamma} \vert \underline{\eta}(t_k)\vert +\left(\beta-e^{-h/\gamma}\right)\vert \underline{u}_k\vert\\
    &\leq \beta \Delta_k\leq \Delta_{k+1}.
\end{align*}
In view of \eqref{eqn_separate_system_4}, the evolution of the acceleration of the system \eqref{eqn_separate_system} during one sampling period can be given by
\begin{align*}
    \underline{\mu_d}&=e^{-(t-t_k)/\gamma}\underline{\mu_d}(t_k)+e^{-(t-t_k)/\gamma}\int_{t_k}^t e^{(\tau-t_k)/\gamma}\frac{\alpha}{\gamma}\underline{\eta}(\tau)d\tau.
\end{align*}
Applying the absolute value to both sides yields that
\begin{align*}
    \vert\underline{\mu_d}\vert&\leq e^{-(t-t_k)/\gamma}\vert\underline{\mu_d}(t_k)\vert\\
    &\quad\quad+e^{-(t-t_k)/\gamma}\int_{t_k}^t e^{-(\tau-t_k)/\gamma}\left\vert\frac{\alpha}{\gamma}\underline{\eta}(\tau)\right\vert d\tau\\
    &\leq  e^{-(t-t_k)/\gamma}\vert \underline{\mu_d}(t_k)\vert+\alpha\left(1-e^{-(t-t_k)/\gamma}\right)\vert  \bar{\underline{\eta}}_k \vert
    \end{align*}
where $\vert\bar{\underline{\eta}}_k\vert=\sup_{t\in[t_k,t_{k+1})}\vert\underline{\eta}(t)\vert\leq \Delta_k$. Following a similar approach to the proof for jerk $\underline{\eta}(t)$, an analogous conclusion can be drawn for acceleration that $\vert\underline{\mu_d}(t)\vert\leq \Delta_k$ for all $t\in[t_k,t_{k+1})$ and $\vert\underline{\mu_d}(t_{k+1})\vert\leq \Delta_{k+1}$.

This completes the proof.

\subsection{Proof of Theorem 1}\label{apx_thm1}
From \cite{rawlings2020model}, to establish the stability, it suffices to prove that for any $x_k$, there exists a feasible local control law $\kappa(x_k)$ such that
\begin{subequations}\label{eqn_stable_assum}
    \begin{align}
        \vert \kappa(x_k)\vert &\leq \Delta_k\\
        V_f(Ax_k+B\kappa(x_k))-V_f(x_k)&\leq -l(x_k,\kappa(x_k)).
    \end{align}    
\end{subequations}
In view of \eqref{eqn_coordinate_transformation}, \eqref{eqn_sub_coordinate_transformation}, and \eqref{eqn_MPCcost}, one has $V_f(x_k)=\Theta(\|\hat{x}_{1,k}\|_W^2 +\hat{x}_{2,k}^2), l(x_k,\kappa(x_k))=\|\hat{x}_{1,k}\|^2+\Psi_\Delta(\hat{x}_{2,k})+\kappa(x_k)^2$, and 
\begin{align}
    &\frac{1}{\Theta}(V_f(Ax_k+B\kappa(x_k))-V_f(x_k))\notag\\
    = & \|\hat{A}_1 \hat{x}_{1,k}+\hat{B}_1\kappa(x_k)\|_W^2-\|\hat{x}_{1,k}\|_W^2\notag\\
    &+(\hat{x}_{2,k}+\hat{b}\kappa(x_k))^2-\hat{x}_{2,k}^2\notag\\
     =&\|\hat{x}_{1,k}\|_{\hat{A}_1\T W \hat{A}_1 -W}^2+ 2\hat{x}_{1,k}\T\hat{A}_1\T W\hat{B}_1\kappa(x_k)\notag\\
    &+2\hat{b}\hat{x}_{2,k}\kappa(x_k)+(\hat{b}^2+\|\hat{B}_1\|_W^2)\kappa(x_k)^2\notag\\
    \leq &-\|\hat{x}_{1,k}\|_{Q-\varepsilon\hat{A}_1\T W W\T\hat{A}_1}^2+2\hat{b}\hat{x}_{2,k}\kappa(x_k)\notag\\
    &+(\hat{b}^2+\|\hat{B}_1\|_W^2+\frac{1}{\varepsilon}\|\hat{B}_1\|^2)\kappa(x_k)^2\notag\\
    =&-\|\hat{x}_{1,k}\|_{Q-\varepsilon\hat{A}_1\T W W\T\hat{A}_1}^2+(2\hat{b}\hat{x}_{2,k}+\Gamma\kappa(x_k) ) \kappa(x_k)\notag\\
    &-\vert \hat{b}\vert \kappa(x_k)^2\label{eqn_ter-ter}
\end{align}
where the inequality follows from Young's inequality with $\varepsilon$ chosen as \eqref{eqn_MPC_para}. Choosing the following local control law  
\begin{align}\label{eqn_localcontrollaw}
    \kappa(x_k)=-\sigma_\Delta ({\hat{b}\hat{x}_{2,k}}{\Gamma}^{-1})
\end{align}
and applying the first inequality in Lemma \ref{lem_1}, one can further rewrite \eqref{eqn_ter-ter} as
\begin{align*}
    &\frac{1}{\Theta}(V_f(Ax_k+B\kappa(x_k))-V_f(x_k))\notag\\
    \leq&-\frac{1}{2}\lambda_{\min}(Q)\|\hat{x}_{1,k}\|^2-\vert \hat{b}\vert \kappa(x_k)^2-\vert \hat{b}\vert \hat{x}_{2,k}\sigma_\Delta(\vert \hat{b}\vert \hat{x}_{2,k}\Gamma^{-1})\notag\\
    \leq&-\frac{1}{2}\lambda_{\min}(Q)\|\hat{x}_{1,k}\|^2-\vert \hat{b}\vert \kappa(x_k)^2-\hat{b}^2\Gamma^{-1}\hat{x}_{2,k}\sigma_\Delta(\hat{x}_{2,k})\notag\\
    \leq&-\frac{1}{2}\lambda_{\min}(Q)\|\hat{x}_{1,k}\|^2-\vert \hat{b}\vert \kappa(x_k)^2-\hat{b}^2\Gamma^{-1}\Psi_\Delta(\hat{x}_{2,k})  \notag\\
    \leq&-\frac{1}{\Theta}\left(\|\hat{x}_{1,k}\|^2+\Psi_\Delta(\hat{x}_{2,k})+\kappa(x_k)^2\right)=-\frac{1}{\Theta}l(x_k,\kappa(x_k))
\end{align*}
where the second inequality follows from the second inequality in Lemma \ref{lem_1}, and the third inequality follows from the third inequality in Lemma \ref{lem_1}. Since the local control \eqref{eqn_localcontrollaw} is feasible for all $k\in\mathbb{N}$ by construction, the stability condition \eqref{eqn_stable_assum} holds uniformly with respect to the sampling index. Therefore, the origin $x_k=0$ of the discrete-time closed-loop system \eqref{eqn_discrete_system}, \eqref{eqn_MPCfeedback} is UGAS. This completes the proof.

\subsection{Proof of Lemma 5}\label{apx_prop1}
The MPC strategy \eqref{eqn_MPClaw} yields the optimal control sequence $\bu^*(x_k)=\{u^*_{0\vert k},u^*_{1\vert k},\cdots,u^*_{N-1\vert k}\}$, and the optimal cost $J_N^*(x_k)=J(x_k,\bu^*(x))$. On one hand,  one has $J_N^*(x_k) =J(x_k,\bu^*(x)) \geq l(x_k,u_{0\vert k}^*)\geq (u_{0\vert k}^*)^2=\underline{u}_k^2$. On the other hand, consider another control sequence $\tilde{\bu}(x_k)=\{\kappa(x_{0\vert k}),\kappa(\tilde{x}_{1\vert k}),\cdots,\kappa(\tilde{x}_{N-1\vert k})\}$, and the corresponding state prediction sequence $\tilde{\mathbf{x}}=\{\tilde{x}_{0\vert k},\tilde{x}_{1\vert k},\cdots,\tilde{x}_{N\vert k}\}$ generated by applying $\tilde{\bu}$ to the system dynamics. Choosing the control law $\kappa(\cdot)$ as \eqref{eqn_localcontrollaw}, then \eqref{eqn_stable_assum} is satisfied, which implies
\begin{align}\label{eqn_J*leq}
    J_N^*(x_k)&\leq J(x_k,\tilde{\bu}(x_k))\notag\\
    &\textstyle =\sum_{i=0}^{N-1}l(\tilde{x}_{i\vert k},\kappa(\tilde{x}_{i\vert k}))+V_f(\tilde{x}_{N\vert k})\notag\\ 
    &\textstyle =\sum_{i=0}^{N-2}l(\tilde{x}_{i\vert k},\kappa(\tilde{x}_{i\vert k}))\notag\\
    &\textstyle ~~~~+l(\tilde{x}_{N-1\vert k},\kappa(\tilde{x}_{N-1\vert k}))+V_f(\tilde{x}_{N\vert k})\notag\\
    &\textstyle \leq \sum_{i=0}^{N-2}l(\tilde{x}_{i\vert k},\kappa(\tilde{x}_{i\vert k}))+V_f(\tilde{x}_{N-1\vert k})\notag\\
    &\textstyle \leq V_f(x_k)\leq \Theta \lambda_{\max}(M)\|x_k\|^2.
\end{align}
It follows that $\underline{u}_k^2\leq \Theta\lambda_{\max}(M)\| x_k\|^2 \leq c_3^2 \|x_k\|^2 $ with $c_3=\sqrt{\Theta\lambda_{\max}(M)}>0$. This completes the proof.

\bibliographystyle{IEEEtran}
\bibliography{MPCbib}
\end{document}